\documentclass[a4paper,12pt]{amsart}
\usepackage[utf8]{inputenc}
\usepackage[T1]{fontenc}

\usepackage{xcolor}
\usepackage{url,hyperref}
\usepackage{listings}
\usepackage{spatialpart}
\usepackage{amsmath,amsthm,amssymb}
\usepackage{algorithm}
\usepackage[noend]{algpseudocode}
\usepackage{cleveref}

\theoremstyle{definition}
\newtheorem{definition}{Definition}
\newtheorem{example}[definition]{Example}
\newtheorem{remark}[definition]{Remark}

\theoremstyle{plain}
\newtheorem{proposition}[definition]{Proposition}
\newtheorem{lemma}[definition]{Lemma}
\newtheorem{theorem}[definition]{Theorem}
\newtheorem{corollary}[definition]{Corollary}

\newcommand{\N}{{\mathbb{N}}}
\newcommand{\Z}{{\mathbb{Z}}}

\newcommand{\C}{{\mathbb{C}}}
\newcommand{\F}{{\mathbb{F}}}
\newcommand{\CC}{{\mathcal{C}}}

\newcommand{\OO}{{\mathcal{O}}}
\newcommand{\PP}{{\mathcal{P}}}
\newcommand{\VV}{{\mathcal{V}}}

\newcommand{\abs}[1]{{\left\lvert#1\right\rvert}}

\DeclareMathOperator{\End}{End}

\newpartitionstyle{2d}{
    aszim   = 0, 
    elev    = 0,
    aspectX = 0.6
}
\newpartitionstyle{3d}{
    aspectX = 0.8,
    aspectZ = 1.2
}

\def\partPair{\partition[2d, scale=1]{%
    \line{0}{0}{0}{0}{0.5}{0}
    \line{0}{0.5}{0}{1}{0.5}{0}
    \line{1}{0.5}{0}{1}{0}{0}
    \point{0}{0}{0}{white}
    \point{1}{0}{0}{white}
}}

\def\partId{\partition[2d, scale=1]{%
    \line{0}{0}{0}{0}{1}{0}
    \point{0}{0}{0}{white}
    \point{0}{1}{0}{white}
}}

\definecolor{juliablue}{rgb}{0.12, 0.47, 0.71}    
\definecolor{juliagreen}{rgb}{0.2, 0.63, 0.17}    
\definecolor{juliagray}{rgb}{0.5, 0.5, 0.5}       
\definecolor{juliabracket}{rgb}{0.55, 0.34, 0.72} 
\definecolor{juliateal}{rgb}{0.0, 0.5, 0.5}       
\definecolor{juliared}{rgb}{0.84, 0.15, 0.16}     
\definecolor{juliaorange}{rgb}{0.9, 0.47, 0.13}    

\begin{document}

\title{Algorithmic Problems in Categories of Partitions}

\author{Nicolas Faroß}
\address{Department of Mathematics, Saarland University, Saarbr\"ucken, Germany}
\email{faross@math.uni-sb.de}
\author{Sebastian Volz}
\address{Saarland University, Saarbr\"ucken, Germany}
\email{s8sevolz@stud.uni-saarland.de}

\date{\today}
\thanks{Both authors thank their supervisor Moritz Weber for many helpful comments and suggestions.
This work is based on parts of the first author's PhD thesis as well as the second author's Bachelor's thesis.
This work has been supported by the SFB-TRR 195 \emph{Symbolic Tools in Mathematics and their Application} of the German Research
Foundation (DFG), project number  286237555.}

\begin{abstract}
Categories of partitions are combinatorial structures
arising from the representation theory of certain compact quantum groups and are linked to classical diagram algebras such as the Temperley-Lieb algebra.
In this paper, we present efficient algorithms and data-structures for partitions of sets and their corresponding category operations, including a concrete implementation in the computer algebra system OSCAR.
Moreover, we show that there exists a category of partitions for which the natural computational problems 
of deciding membership of a given partition as well as counting partitions of a given size are algorithmically undecidable.
\end{abstract}

\maketitle

\section{Introduction}

Partitions are combinatorial objects given by a row of upper and a row of lower points that are partitioned into disjoint subsets by strings, e.g.\
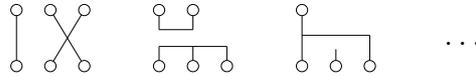
\[
    \partition[2d]{%
    \line{0}{0}{0}{0}{1}{0}
    \line{1}{0}{0}{2}{1}{0}
    \line{2}{0}{0}{1}{1}{0}
    \point{0}{0}{0}{white}
    \point{0}{1}{0}{white}
    \point{1}{0}{0}{white}
    \point{1}{1}{0}{white}
    \point{2}{0}{0}{white}
    \point{2}{1}{0}{white}
    }
    \qquad
    \partition[2d]{%
    \line{0}{0}{0}{0}{0.35}{0}
    \line{0}{0.35}{0}{2}{0.35}{0}
    \line{1}{0.35}{0}{1}{0}{0}
    \line{2}{0.35}{0}{2}{0}{0}
    \line{0}{0.65}{0}{1}{0.65}{0}
    \line{0}{0.65}{0}{0}{1}{0}
    \line{1}{0.65}{0}{1}{1}{0}
    \point{0}{0}{0}{white}
    \point{0}{1}{0}{white}
    \point{1}{0}{0}{white}
    \point{1}{1}{0}{white}
    \point{2}{0}{0}{white}
    }
    \qquad
    \partition[2d]{%
    \line{0}{0}{0}{0}{1}{0}
    \line{0}{0.55}{0}{2}{0.55}{0}
    \line{1}{0}{0}{1}{0.3}{0}
    \line{2}{0.55}{0}{2}{0}{0}
    \point{0}{0}{0}{white}
    \point{0}{1}{0}{white}
    \point{1}{0}{0}{white}
    \point{2}{0}{0}{white}
    }
    \qquad
    \dots
\]
Given a set of partitions, we can construct new partitions by applying various operations like a tensor product, an involution or a composition. 
A category of partitions in the sense of~\cite{weber17} is a set of partitions that is closed under these operations and contains the base partitions $\partId$ and $\partPair$.

Partitions appear naturally in the definition diagram algebras such as the Brauer algebra~\cite{brauer37} and are linked to the representation theory of certain classical groups via Schur-Weyl duality~\cite{fulton04}. In~\cite{banica09}, Banica-Speicher introduced so-called easy quantum groups, a class of compact quantum groups defined from categories of partitions. These quantum groups extend Schur-Weyl duality to a larger class of diagram algebras including the Temperley-Lieb algebra~\cite{temperly71}, which has many applications from knot theory to quantum mechanics~\cite{abramsky07}.

Interest in easy quantum groups has lead to the study of categories of 
partitions in~\cite{banica10a,weber13} with a full classification completed in~\cite{raum16a}. Moreover, there exist various generalizations like colored partitions~\cite{tarrago18}
and three-dimensional spatial partitions~\cite{cebron16} for which the classification is still ongoing.
We refer to~\cite{weber17, freslon23b}
for an introduction to compact quantum groups defined by categories of partitions and to~\cite{levandovskyy22,faross24a,corey23} for more information on the recent use of computational methods in the setting of compact quantum groups.

\subsection{Main results}
In the first part, we present efficient algorithms and data-structures for partitions, their category operations as well as generalizations to colored and spatial partitions. Moreover, we provide a concrete implementation of these algorithms and data-structures in the computer algebra system OSCAR~\cite{OSCAR}. This enables the study of easy quantum groups via computational methods, possibly leading to advances in the ongoing classification of categories of colored and spatial partitions. In addition, our algorithm for computing the composition of two partitions runs in quasi-linear time, which improves on the worst-case complexity of existing implementations, for example in~\cite{sagemath}.
 
In the second part, we study the computational problems of deciding membership of a given partition as well as counting partitions of a given size in categories of partitions. These problems appear naturally when studying categories of partitions and their corresponding quantum groups and can be easily solved for many concrete examples. However, our main theorem now shows that there exists a concrete category of partitions for which these problems are algorithmically undecidable.

\begin{theorem}[\Cref{thm:membership-undecidable}, \Cref{corr:counting-undecidable}]
There exists a recursively enumerable category of partitions $\CC$ for which the following problems are algorithmically undecidable: 
\begin{enumerate}
  \item Decide if $p \in \CC$ for a given partition $p$.
  \item Determine the number of partitions in $\CC$ with $k$ upper and $\ell$  lower points for given $k, \ell \in \N$.
\end{enumerate}
\end{theorem}

Here, a category of partitions is called recursively enumerable if its elements can be enumerated by a Turing machine, and a problem is called algorithmically undecidable if there exists no Turing machine solving the corresponding problem for all possible input values in finite time.

\subsection{Structure of the paper}

We begin with some preliminaries in \Cref{sec:cat-of-partitions}, where we introduce various types of partitions, their basic operations as well as categories of partitions.
Then, we present efficient algorithms and data-structures 
for partitions and their operations in \Cref{sec:algos-structs}, including a short demo of our OSCAR implementation in \Cref{sec:oscar-impl}.
Next, we study algorithmic problems for categories of partitions in \Cref{sec:undeciability}, where we prove that the problems of deciding membership of a partition and counting partitions of a given size are generally undecidable.
Finally, we finish with some concluding remarks and open questions in \Cref{sec:conclusion}.

\numberwithin{definition}{section}

\section{Preliminaries}\label{sec:cat-of-partitions}

\subsection{Partitions}\label{sec:def_partition}

We begin by introducing basic partitions, which are the main combinatorial objects used in the following sections to define categories of partitions.

\begin{definition}[Partition]\label{def:partition}
Let $k, \ell \in \mathbb{N}$. A \emph{partition} with $k$ upper and $\ell$ upper points is given by a decomposition of the set $\{1, \dots, k + \ell\}$ into disjoint subsets called blocks.
\end{definition}

In the following, we denote by $\PP$ the set of all partitions, and by $\PP(k, \ell)$ the set of all partitions with $k$ upper and $\ell$ lower points.
Given a partition $p \in \PP(k, \ell)$, we can visualize it as a string-diagram as follows.
First, we draw a row of $k$ upper points labeled by $1, \dots, k$ and a row of $\ell$ lower points
labeled by $k+1, \dots, k+\ell$.
Then, we connect points by strings if their corresponding labels are in the same block in $p$. Finally, we omit the labels and visualize each point as simple dot.

\begin{example}\label{ex:draw-partition}
Consider the partition $p \in \PP(4, 3)$ with underlying decomposition given by $\{\{1, 6\}, \{2, 5\}, \{3, 4, 7\}\}$. Then we can visualize $p$ as follows:
\[
\begin{array}{c}
\{\{1, 6\}, \{2, 5\}, \\
\{3, 4, 7\}\}
\end{array}
\quad \cong \quad
\partition[2d]{%
    \line{1}{0}{0}{0}{1}{0}
    \line{0}{0}{0}{1}{1}{0}
    \line{2}{0}{0}{2}{1}{0}
    \line{2}{0.5}{0}{3}{0.5}{0}
    \line{3}{0.5}{0}{3}{1}{0}
    \label{0}{1.2}{0}{\Huge 1}
    \label{1}{1.2}{0}{\Huge 2}
    \label{2}{1.2}{0}{\Huge 3}
    \label{3}{1.2}{0}{\Huge 4}
    \label{0}{-0.2}{0}{\Huge 5}
    \label{1}{-0.2}{0}{\Huge 6}
    \label{2}{-0.2}{0}{\Huge 7}
}
\quad \cong \quad
\partition[2d]{%
    \line{1}{0}{0}{0}{1}{0}
    \line{0}{0}{0}{1}{1}{0}
    \line{2}{0}{0}{2}{1}{0}
    \line{2}{0.5}{0}{3}{0.5}{0}
    \line{3}{0.5}{0}{3}{1}{0}
    \point{1}{0}{0}{white}
    \point{0}{1}{0}{white}
    \point{1}{1}{0}{white}
    \point{0}{0}{0}{white}
    \point{2}{0}{0}{white}
    \point{2}{1}{0}{white}
    \point{3}{1}{0}{white}
}
\]
\end{example}

\begin{remark}
Note that we distinguish partitions with a different number of upper and lower points, e.g.\ the following 
partitions all have the underlying decomposition \{\{1, 2\}\}
but are distinct:
\[
    \partPair \in \PP(0, 2), \qquad
    \partId \in \PP(1, 1), \qquad
    \partition[2d]{%
        \line{0}{1}{0}{0}{0.5}{0}
        \line{0}{0.5}{0}{1}{0.5}{0}
        \line{1}{0.5}{0}{1}{1}{0}
        \point{0}{1}{0}{white}
        \point{1}{1}{0}{white}
    } \in \PP(2, 0).
\]
\end{remark}

\subsection{Operations on partitions}\label{sec:op_partition}

Given a set of partitions, we can construct new partitions using the following basic operations called tensor product, involution and composition.

\begin{definition}[Tensor product] Let $p \in \PP(k_1, \ell_1)$ and $q \in \PP(k_2, \ell_2)$. Then their \textit{tensor product} $p \otimes q \in \PP(k_1 + k_2, \ell_1 + \ell_2)$ is obtained by concatenating $p$ and $q$ horizontally, e.g.\
\[
\partition[2d]{%
    \line{1}{0}{0}{0}{1}{0}
    \line{0}{0}{0}{1}{1}{0}
    \point{1}{0}{0}{white}
    \point{0}{1}{0}{white}
    \point{1}{1}{0}{white}
    \point{0}{0}{0}{white}
}
\quad \otimes \quad
\partition[2d]{%
    \line{2}{0}{0}{2}{1}{0}
    \line{2}{0.5}{0}{3}{0.5}{0}
    \line{3}{0.5}{0}{3}{1}{0}
    \point{2}{0}{0}{white}
    \point{2}{1}{0}{white}
    \point{3}{1}{0}{white}
}
\quad = \quad
\partition[2d]{%
    \line{1}{0}{0}{0}{1}{0}
    \line{0}{0}{0}{1}{1}{0}
    \line{2}{0}{0}{2}{1}{0}
    \line{2}{0.5}{0}{3}{0.5}{0}
    \line{3}{0.5}{0}{3}{1}{0}
    \point{1}{0}{0}{white}
    \point{0}{1}{0}{white}
    \point{1}{1}{0}{white}
    \point{0}{0}{0}{white}
    \point{2}{0}{0}{white}
    \point{2}{1}{0}{white}
    \point{3}{1}{0}{white}
}.
\]
\end{definition}

\begin{definition}[Involution]
Let $p \in \PP(k, \ell)$. Then the \textit{involution} $p^* \in \PP(\ell, k)$ is obtained by reflecting the upper and lower points along a horizontal axis, e.g.\
\[
\left(~
\partition[2d]{%
    \line{2}{0}{0}{2}{1}{0}
    \line{2}{0.35}{0}{3}{0.35}{0}
    \line{3}{0.35}{0}{3}{0}{0}
    \line{4}{0}{0}{4}{0.35}{0}
    \line{3}{1}{0}{3}{0.65}{0}
    \line{3}{0.65}{0}{4}{0.65}{0}
    \line{4}{0.65}{0}{4}{1}{0}
    \point{2}{0}{0}{white}
    \point{2}{1}{0}{white}
    \point{3}{0}{0}{white}
    \point{4}{0}{0}{white}
    \point{3}{1}{0}{white}
    \point{4}{1}{0}{white}
}
~\right)^{*}
\quad = \quad
\partition[2d]{%
    \line{2}{1}{0}{2}{0}{0}
    \line{2}{0.65}{0}{3}{0.65}{0}
    \line{3}{0.65}{0}{3}{1}{0}
    \line{4}{1}{0}{4}{0.65}{0}
    \line{3}{0}{0}{3}{0.35}{0}
    \line{3}{0.35}{0}{4}{0.35}{0}
    \line{4}{0.35}{0}{4}{0}{0}
    \point{2}{1}{0}{white}
    \point{2}{0}{0}{white}
    \point{3}{1}{0}{white}
    \point{4}{1}{0}{white}
    \point{3}{0}{0}{white}
    \point{4}{0}{0}{white}
}.
\]
\end{definition} 

\begin{definition}[Composition]\label{def: comp}
Let $p \in P(\ell, m)$ and $q \in P(k, \ell)$. Then the \textit{composition} $pq \in P(k, m)$ is obtained placing $q$ above $p$ and identifying each lower point of $q$ with the respective upper point of $p$. Then, these intermediate points are removed such that only the upper points of $q$ and the lower points of $p$ are left, e.g.\
\[
\partition[2d]{%
    \line{0}{0}{0}{0}{1}{0}
    \line{1}{1}{0}{1}{0}{0}
    \line{1}{0.5}{0}{2}{0.5}{0}
    \line{2}{0.5}{0}{2}{1}{0}
    \point{0}{0}{0}{white}
    \point{0}{1}{0}{white}
    \point{1}{1}{0}{white}
    \point{2}{1}{0}{white}
    \point{1}{0}{0}{white}
}
\quad \cdot \quad
\partition[2d]{%
    \line{0}{0}{0}{0}{0.5}{0}
    \line{0}{0.5}{0}{1}{0.5}{0}
    \line{1}{0.5}{0}{1}{0}{0}
    \line{2}{0}{0}{2}{1}{0}
    \point{0}{0}{0}{white}
    \point{1}{0}{0}{white}
    \point{2}{0}{0}{white}
    \point{2}{1}{0}{white}
}
\quad = \quad
\partition[2d]{%
    \line{0}{0}{0}{0}{1}{0}
    \line{1}{1}{0}{1}{0}{0}
    \line{1}{0.5}{0}{2}{0.5}{0}
    \line{2}{0.5}{0}{2}{1}{0}
    \line{0}{1}{0}{0}{1.5}{0}
    \line{0}{1.5}{0}{1}{1.5}{0}
    \line{1}{1.5}{0}{1}{1}{0}
    \line{2}{1}{0}{2}{2}{0}
    \point{0}{0}{0}{white}
    \point{0}{1}{0}{white}
    \point{1}{1}{0}{white}
    \point{2}{1}{0}{white}
    \point{1}{0}{0}{white}
    \point{2}{2}{0}{white}
}
\quad = \quad
\partition[2d]{%
    \line{1}{1}{0}{1}{0}{0}
    \line{2}{0}{0}{2}{0.5}{0}
    \line{1}{0.5}{0}{2}{0.5}{0}
    \point{1}{1}{0}{white}
    \point{1}{0}{0}{white}
    \point{2}{0}{0}{white}
}.
\]
\end{definition}

\begin{remark}
Note that the composition of two partitions $p$ and $q$ is only defined if the number of lower points of $q$ equals the number of upper points of $p$.
Moreover, when removing the intermediate points, any isolated components that are not connected to the top or bottom row are simply discarded. In particular, they do not contribute a loop factor as in the case of diagram algebras.
\end{remark}

In addition to the previous operations, it is often useful to additionally consider rotations and vertical reflections. We refer to~\cite{raum16a} for more information the use of these operations and the combintorics of partitions in general.

\begin{definition}[Rotations]
Let $p \in \PP(k, \ell)$ with $k > 0$. Then the \textit{(top-left) rotation} $p^{\mathit{rot}} \in \PP(k - 1, \ell + 1)$ is obtained by rotating the left-most upper point to the left-most position in the lower row while preserving the block structure, e.g.\ 
\[
\left(~
\partition[2d]{%
    \line{1}{0}{0}{1}{1}{0}
    \line{0}{0.65}{0}{1}{0.65}{0}
    \line{0}{0.65}{0}{0}{1}{0}
    \line{0}{0}{0}{0}{0.35}{0}
    \point{0}{0}{0}{white}
    \point{0}{1}{0}{white}
    \point{1}{1}{0}{white}
    \point{1}{0}{0}{white}
}
~\right)^{\mathit{rot}}
\quad = \quad
\partition[2d]{%
    \line{2}{0}{0}{2}{1}{0}
    \line{2}{0.6}{0}{4}{0.6}{0}
    \line{4}{0.6}{0}{4}{0}{0}
    \line{3}{0}{0}{3}{0.3}{0}
    \point{2}{0}{0}{white}
    \point{2}{1}{0}{white}
    \point{3}{0}{0}{white}
    \point{4}{0}{0}{white}
}.
\]

Analogously, we can also define a rotation of the left-most lower point to the upper row as well as rotations of the right-most upper point to the lower row and vice versa.
\end{definition} 

\begin{definition}[Vertical reflection]
Let $p \in \PP(k, \ell)$. Then the  \textit{vertical reflection} $p^\leftrightarrow \in \PP(k, \ell)$ is obtained by reflecting the upper and lower points along the vertical axis, e.g.\
\[
\left(~
\partition[2d]{%
    \line{2}{0}{0}{2}{1}{0}
    \line{2}{0.35}{0}{3}{0.35}{0}
    \line{3}{0.35}{0}{3}{0}{0}
    \line{4}{0}{0}{4}{0.35}{0}
    \line{3}{1}{0}{3}{0.65}{0}
    \line{3}{0.65}{0}{4}{0.65}{0}
    \line{4}{0.65}{0}{4}{1}{0}
    \point{2}{0}{0}{white}
    \point{2}{1}{0}{white}
    \point{3}{0}{0}{white}
    \point{4}{0}{0}{white}
    \point{3}{1}{0}{white}
    \point{4}{1}{0}{white}
}
~\right)^{\leftrightarrow}
\quad = \quad
\partition[2d]{%
    \line{4}{0}{0}{4}{1}{0}
    \line{3}{0.35}{0}{4}{0.35}{0}
    \line{3}{0.35}{0}{3}{0}{0}
    \line{2}{0}{0}{2}{0.35}{0}
    \line{2}{1}{0}{2}{0.65}{0}
    \line{2}{0.65}{0}{3}{0.65}{0}
    \line{3}{0.65}{0}{3}{1}{0}
    \point{3}{0}{0}{white}
    \point{4}{1}{0}{white}
    \point{4}{0}{0}{white}
    \point{2}{0}{0}{white}
    \point{2}{1}{0}{white}
    \point{3}{1}{0}{white}
}.
\]
\end{definition} 

\subsection{Categories of partitions}\label{sec:categories}

Given partitions and their basic operations, we can finally introduce categories of partitions.

\begin{definition}[Category of partitions]
A \emph{category of partitions} is a subset $\CC \subseteq \PP$ that contains the \emph{base partitions} $\partId \in \PP(1, 1)$ and $\partPair \in \PP(0, 2)$ 
and is closed under tensor products, involutions and compositions.
\end{definition}

As in the case of all partitions, we denote with 
$\CC(k, \ell) := \CC \cap \PP(k, \ell)$ the set of all partitions in a category $\CC$ with $k$ upper and $\ell$ lower points.

\begin{remark}
Note that categories of partitions are automatically closed under rotations and vertical reflections since these operations can be constructed by combining the previous category operations with the base partitions $\partPair$ and $\partId$. See~\cite{raum16a} for further details.
\end{remark}

\begin{example}\label{ex:categories}
Examples of categories of partitions 
include
\begin{enumerate}
\item the set $\PP$ of all partitions.
\item the set $\PP_2$ of all \emph{pair partitions}, i.e.\ partitions with all blocks exactly of size two.
\item the set $\mathcal{B}$ of all \emph{balanced partitions}, i.e.\ partitions with the same number of ``even'' and ``odd'' points in each block. 
\item the set $\mathcal{NC}$ of all \emph{non-crossing} partitions, i.e.\ partitions that can be drawn as a string diagram without crossing lines.
\end{enumerate}
\end{example}

Note that categories of partitions are closed under intersections. Thus, any combination of the previous properties defines again a category of partitions. See also~\cite{raum16a} for further examples. In particular, it is also possible to define categories of partitions via a generating set of partitions.

\begin{definition}[Generating set]
Let $S \subseteq \PP$ be a set of partitions. Then we denote with $\langle S \rangle$ the \textit{category generated by $S$}.
It is the smallest category of partitions containing $S$ and consists of all finite combinations of the base partitions, elements of $S$ and the category operations. Additionally, we call a category $\CC$ \emph{finitely generated} if  $\CC = \langle S \rangle$ for a finite set $S$.
\end{definition}

\begin{example}\label{ex:categories-gens}
Consider again the previous categories of partitions from \Cref{ex:categories}. Then these categories can alternative be defined as follows:
\begin{alignat*}{2}
\PP &= \Big\langle \
\partition[2d]{%
    \line{0}{0}{0}{0}{1}{0}
    \line{1}{0}{0}{1}{0.5}{0}
    \line{0}{0.5}{0}{1}{0.5}{0}
    \point{0}{0}{0}{white}
    \point{0}{1}{0}{white}
    \point{1}{0}{0}{white}
}, \ 
\partition[2d]{
    \line{0}{0}{0}{1}{1}{0}
    \line{1}{0}{0}{0}{1}{0}
    \point{0}{0}{0}{white}
    \point{0}{1}{0}{white}
    \point{1}{0}{0}{white}
    \point{1}{1}{0}{white}
}
\ \Big\rangle, 
\qquad & \PP_2 &= \Big\langle \
\partition[2d]{%
    \line{0}{0}{0}{1}{1}{0}
    \line{1}{0}{0}{0}{1}{0}
    \point{0}{0}{0}{white}
    \point{0}{1}{0}{white}
    \point{1}{0}{0}{white}
    \point{1}{1}{0}{white}
}
\ \Big\rangle, \\
\mathcal{B} &= \Big\langle \
\partition[2d]{%
    \line{0}{0}{0}{0}{0.33}{0}
    \line{0}{1}{0}{0}{0.67}{0}
    \line{1}{0}{0}{1}{0.33}{0}
    \line{1}{1}{0}{1}{0.67}{0}
    \line{0}{0.33}{0}{1}{0.33}{0}
    \line{0}{0.67}{0}{1}{0.67}{0}
    \line{0.5}{0.33}{0}{0.5}{0.67}{0}
    \point{0}{0}{0}{white}
    \point{0}{1}{0}{white}
    \point{1}{0}{0}{white}
    \point{1}{1}{0}{white}
}, \ 
\partition[2d]{
    \line{0}{0}{0}{2}{1}{0}
    \line{1}{0}{0}{1}{1}{0}
    \line{2}{0}{0}{0}{1}{0}
    \point{0}{0}{0}{white}
    \point{0}{1}{0}{white}
    \point{1}{0}{0}{white}
    \point{1}{1}{0}{white}
    \point{2}{0}{0}{white}
    \point{2}{1}{0}{white}
}
\ \Big\rangle, 
\qquad & \mathcal{NC} &= \Big\langle \
\partition[2d]{%
    \line{0}{0}{0}{0}{1}{0}
    \line{1}{0}{0}{1}{0.5}{0}
    \line{0}{0.5}{0}{1}{0.5}{0}
    \point{0}{0}{0}{white}
    \point{0}{1}{0}{white}
    \point{1}{0}{0}{white}
}
\ \Big\rangle.
\end{alignat*}
We refer to~\cite{raum16a} for further information and additional examples of categories of partitions that are not necessarily finitely generated.
\end{example}

\begin{remark}
Let us comment quickly on the connection between categories of partitions and compact quantum groups. Given $n \in \N$ and a category of partitions $\CC$, we can associate a linear operator
$
    T_p \colon {(\C^n)}^{\otimes k} \to {(\C^n)}^{\otimes k} 
$
to every partition $p \in \CC(k, \ell)$. These operators define a concrete $C^*$-tensor category, which corresponds to a compact quantum group $G_n$ with distinguished fundamental representation $u$ via Woronowicz-Tannaka-Krein duality~\cite{woronowicz88}. In the case of all partitions, we obtain for example the classical symmetric group $S_n$, whereas the category $\mathcal{NC}$ of all non-crossing partitions corresponds to Wang's quantum permutation group $S_n^+$~\cite{wang98}.
We refer to~\cite{weber17,freslon23b} for more details on their construction and to~\cite{timmermann08} for the theory compact quantum groups in general.
\end{remark}

\subsection{Generalizations of partitions}\label{sec:generalization}

In addition to the previous partitions, there exists a generalization to colored partitions~\cite{tarrago18}, which yield the class of unitary easy quantum groups~\cite{tarrago17}. 

\begin{definition}[Colored partition]
A \emph{colored partition} is given by a partition $p \in \PP(k, \ell)$ with colorings 
$x \in {\{\circ, \bullet\}}^k$ and $y \in {\{\circ, \bullet\}}^\ell$
of its upper and lower points respectively.
\end{definition}

\begin{example}
Consider again the partition $p \in \PP(4, 3)$
from \Cref{ex:draw-partition}.
By choosing the upper colors ${{\circ}{\bullet}{\bullet}{\circ}}$ and the lower colors ${{\circ}{\bullet}{\circ}}$, we obtain the colored following colored partition:
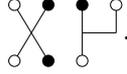
\[
\partition[2d]{%
    \line{1}{0}{0}{0}{1}{0}
    \line{0}{0}{0}{1}{1}{0}
    \line{2}{0}{0}{2}{1}{0}
    \line{2}{0.5}{0}{3}{0.5}{0}
    \line{3}{0.5}{0}{3}{1}{0}
    \point{1}{0}{0}{black}
    \point{0}{1}{0}{white}
    \point{1}{1}{0}{black}
    \point{0}{0}{0}{white}
    \point{2}{0}{0}{white}
    \point{2}{1}{0}{black}
    \point{3}{1}{0}{white}
}.
\]
\end{example}

\begin{remark}
The operations for partitions can be directly extended to colored partitions by additionally requiring that not only the number of upper and lower points but also their colors agree when composing to partitions $p$ and $q$.
Moreover, in categories of colored partitions,
the the base partition $\partId$ and $\partPair$
are replaced by the \emph{colored base partitions} \partition[2d]{%
    \line{0}{0}{0}{0}{1}{0}
    \point{0}{0}{0}{white}
    \point{0}{1}{0}{white}
}, 
\partition[2d]{%
    \line{0}{0}{0}{0}{1}{0}
    \point{0}{0}{0}{black}
    \point{0}{1}{0}{black}
}
and
\partition[2d]{%
    \line{0}{0}{0}{0}{0.5}{0}
    \line{0}{0.5}{0}{1}{0.5}{0}
    \line{1}{0.5}{0}{1}{0}{0}
    \point{0}{0}{0}{white}
    \point{1}{0}{0}{black}
}, 
\partition[2d]{%
    \line{0}{0}{0}{0}{0.5}{0}
    \line{0}{0.5}{0}{1}{0.5}{0}
    \line{1}{0.5}{0}{1}{0}{0}
    \point{0}{0}{0}{black}
    \point{1}{0}{0}{white}
}.
\end{remark}

More recently, Cébron-Weber~\cite{cebron16} also introduced spatial partitions by extending the previous partitions along multiple levels into the third dimension.

\begin{definition}[Spatial partition]
Let $k, \ell, m \in \N$. Then a \emph{spatial partition} with $k$ upper points, $\ell$ lower points and $m$ levels is a decomposition of the set $\{1, \dots, k + \ell\} \times \{1, \dots, m\}$ into disjoint subsets called blocks.
\end{definition}

As in the case of classical partitions, the first component $\{1, \dots, k + \ell\}$ divides a spatial partition into $k$ upper and $\ell$ lower points, whereas the second component $\{1, \dots, m\}$ divides the partition into $m$ levels along the third dimension. In the following, we denote with $\PP^{(m)}$ the set of all spatial partitions on $m$ levels, and with $\PP^{(m)}(k, \ell)$ the subset consisting of all spatial partitions with $k$ upper points and $\ell$ lower points.

\begin{example}
The following example visualizes a spatial partition $p \in P^{(2)}(1,2)$ with one upper point, two lower points and two levels:
\[
\begin{array}{c}
\scriptstyle\{\{(1, 1), (2, 1)\},
\scriptstyle\{(1, 2), (3, 2)\}, \\
\scriptstyle\{(2, 2), (3, 1)\}\}
\end{array}
\quad \cong \quad
\partition[3d]{
  %
  %
  %
  \line{0}{0}{0}{0}{1}{0}
  \line{1}{0}{1}{0}{1}{1}
  \line{0}{0}{1}{0}{0.3}{1}
  \line{0}{0.3}{1}{1}{0.3}{0}
  \line{1}{0.3}{0}{1}{0}{0}
  \label{0}{1}{0}{\Huge (1,1)}
  \label{0}{1}{1}{\Huge (1,2)}
  \label{0}{0}{0}{\Huge (2,1)}
  \label{0}{0}{1}{\Huge (2,2)}
  \label{1}{0}{0}{\Huge (3,1)}
  \label{1}{0}{1}{\Huge (3,2)}
}
\quad \cong \quad
\partition[3d]{
  %
  %
  %
  \line{0}{0}{0}{0}{1}{0}
  \line{1}{0}{1}{0}{1}{1}
  \line{0}{0}{1}{0}{0.3}{1}
  \line{0}{0.3}{1}{1}{0.3}{0}
  \line{1}{0.3}{0}{1}{0.3}{0}
  \line{1}{0.3}{0}{1}{0}{0}
  \point{0}{0}{0}{white}
  \point{1}{0}{0}{white}
  \point{0}{1}{0}{white}
  \point{0}{0}{1}{white}
  \point{1}{0}{1}{white}
  \point{0}{1}{1}{white}
}
\]
\end{example}

\begin{remark}
As before, we can extend the basic operations on partitions to spatial partitions on the same number of levels. 
\begin{enumerate}
\item Tensor product:
\[
  \partition[3d]{%
  \line{0}{0}{0}{0}{1}{0}
  \line{0}{0}{1}{0}{1}{1}
  \point{0}{1}{0}{white}
  \point{0}{1}{1}{white}
  \point{0}{0}{0}{white}
  \point{0}{0}{1}{white}
  }
  \quad \otimes \quad 
  \partition[3d]{%
  \line{0}{0}{0}{0}{0.35}{0}
  \line{0}{0}{1}{0}{0.35}{1}
  \line{0}{0.35}{0}{0}{0.35}{1}
  \point{0}{0}{0}{white}
  \point{0}{0}{1}{white}
  } 
  \quad = \quad 
  \partition[3d]{%
  \line{0}{0}{0}{0}{1}{0}
  \line{0}{0}{1}{0}{1}{1}
  \line{1}{0}{0}{1}{0.35}{0}
  \line{1}{0}{1}{1}{0.35}{1}
  \line{1}{0.35}{0}{1}{0.35}{1}
  \point{0}{1}{0}{white}
  \point{0}{1}{1}{white}
  \point{0}{0}{0}{white}
  \point{0}{0}{1}{white}
  \point{1}{0}{0}{white}
  \point{1}{0}{1}{white}
  }
\]
\item Involution:
\[
  {\left( \ \partition[3d]{%
    \line{0}{0.65}{0}{0}{1}{0}
    \line{0}{0}{0}{0}{0.35}{0}
    \line{0}{0.35}{0}{0}{0.35}{1}
    \line{0}{0}{1}{0}{1}{1}
    \line{1}{0}{1}{1}{0.35}{1}
    \line{1}{0}{0}{1}{0.35}{0}
    \line{1}{0.35}{1}{1}{0.35}{0}
    \point{0}{0}{0}{white}
    \point{0}{0}{1}{white}
    \point{1}{0}{0}{white}
    \point{1}{0}{1}{white}
    \point{0}{1}{0}{white}
    \point{0}{1}{1}{white}} \ \right)}^*
    \quad = \quad
    \partition[3d]{%
    \line{0}{0.35}{0}{0}{0}{0}
    \line{0}{1}{0}{0}{0.65}{0}
    \line{0}{0.65}{0}{0}{0.65}{1}
    \line{0}{1}{1}{0}{0}{1}
    \line{1}{1}{1}{1}{0.65}{1}
    \line{1}{1}{0}{1}{0.65}{0}
    \line{1}{0.65}{1}{1}{0.65}{0}
    \point{0}{1}{0}{white}
    \point{0}{1}{1}{white}
    \point{1}{1}{0}{white}
    \point{1}{1}{1}{white}
    \point{0}{0}{0}{white}
    \point{0}{0}{1}{white}}
\]
\item Composition:
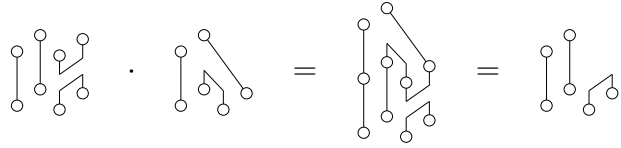
\[
  \partition[3d]{%
  \line{0}{0}{0}{0}{1}{0}
  \line{0}{0}{1}{0}{1}{1}
  \line{1}{0}{0}{1}{0.35}{0}
  \line{1}{0}{1}{1}{0.35}{1}
  \line{1}{0.35}{0}{1}{0.35}{1}
  \line{1}{1}{0}{1}{0.65}{0}
  \line{1}{1}{1}{1}{0.65}{1}
  \line{1}{0.65}{0}{1}{0.65}{1}
  \point{0}{1}{0}{white}
  \point{0}{1}{1}{white}
  \point{1}{1}{0}{white}
  \point{1}{1}{1}{white}
  \point{0}{0}{0}{white}
  \point{0}{0}{1}{white}
  \point{1}{0}{0}{white}
  \point{1}{0}{1}{white}}
  \quad \cdot \quad
  \partition[3d]{%
  \line{0}{0}{0}{0}{1}{0}
  \line{0}{1}{1}{1}{0}{1}
  \line{0}{0}{1}{0}{0.35}{1}
  \line{1}{0}{0}{1}{0.35}{0}
  \line{0}{0.35}{1}{1}{0.35}{0}
  \point{0}{1}{0}{white}
  \point{0}{1}{1}{white}
  \point{0}{0}{0}{white}
  \point{0}{0}{1}{white}
  \point{1}{0}{0}{white}
  \point{1}{0}{1}{white}}
  \quad = \quad 
  \partition[3d]{%
  \line{0}{0}{0}{0}{1}{0}
  \line{0}{0}{1}{0}{1}{1}
  \line{1}{0}{0}{1}{0.35}{0}
  \line{1}{0}{1}{1}{0.35}{1}
  \line{1}{0.35}{0}{1}{0.35}{1}
  \line{1}{1}{0}{1}{0.65}{0}
  \line{1}{1}{1}{1}{0.65}{1}
  \line{1}{0.65}{0}{1}{0.65}{1}
  \line{0}{1}{0}{0}{2}{0}
  \line{0}{2}{1}{1}{1}{1}
  \line{0}{1}{1}{0}{1.35}{1}
  \line{1}{1}{0}{1}{1.35}{0}
  \line{0}{1.35}{1}{1}{1.35}{0}
  \point{0}{1}{0}{white}
  \point{0}{1}{1}{white}
  \point{1}{1}{0}{white}
  \point{1}{1}{1}{white}
  \point{0}{0}{0}{white}
  \point{0}{0}{1}{white}
  \point{1}{0}{0}{white}
  \point{1}{0}{1}{white}
  \point{0}{2}{0}{white}
  \point{0}{2}{1}{white}
  }
  \quad = \quad 
  \partition[3d]{%
  \line{0}{0}{0}{0}{1}{0}
  \line{0}{0}{1}{0}{1}{1}
  \line{1}{0}{0}{1}{0.35}{0}
  \line{1}{0}{1}{1}{0.35}{1}
  \line{1}{0.35}{0}{1}{0.35}{1}
  \point{0}{1}{0}{white}
  \point{0}{1}{1}{white}
  \point{0}{0}{0}{white}
  \point{0}{0}{1}{white}
  \point{1}{0}{0}{white}
  \point{1}{0}{1}{white}
  }
\]
\end{enumerate}
Moreover, a category of spatial partitions is a subset $\CC \subseteq \PP^{(m)}$ that is closed under these operations and contains the spatial base partitions $\partId^{(m)} \in \PP^{(m)}(1,1)$ and $\partPair^{(m)} \in \PP^{(m)}(0,2)$.
Here, $p^{(m)} \in \PP^{(m)}$ denotes the spatial partition that is obtained by placing a copy of a partition $p \in \PP$ on each of the levels.
\end{remark}

\begin{remark}
While the classification of categories of partitions was completed in~\cite{raum16b}, it is still ongoing in the case of colored and spatial partitions. For more information on these cases, we refer for example to~\cite{gromada18, tarrago18}
and the recent work by Mang-Weber~\cite{mang21a,mang21b}
in the case colored partitions, and to~\cite{cebron16} in the case of spatial partition.
\end{remark}

\section{Algorithms and data-structures for partitions}\label{sec:algos-structs}
In the following, we present efficient algorithms and data-structures for partitions, their basic operators as well as their generalizations to colored and spatial partitions. 
Additionally, these algorithms and data-structures have been
implemented in the computer algebra system OSCAR~\cite{OSCAR}. We refer to \Cref{sec:oscar-impl} for a short demo.

\subsection{Partitions}

We begin by introducing a data-structure for partitions in the sense of \Cref{def:partition}.

\begin{definition}[Partition structure]\label{def:part-struct}
A partition $p \in \PP$ can be represented by a tuple $(k, \ell, b)$, where
$k$ denotes the number of upper points, $\ell$ denotes the number of lower points and 
$b \in \N^{k+\ell}$ encodes the block structure such that two points $1 \leq i, j \leq k+\ell$ are in the same block if and only if $b_i = b_j$.
\end{definition}

\begin{example}\label{ex:part-datastruct}
Consider for example the partition $p \in \PP$ represented by the tuple
\[
  \big(3, \, 4, \, (1, 1, 2, 1, 3, 3, 2)\big).
\]
Then $p$ corresponds to the following string diagram:
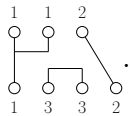
\[
  {\partition[2d]{
  \line{1}{0}{0}{1}{0.35}{0}
  \line{2}{0}{0}{2}{0.35}{0} 
  \line{1}{0.35}{0}{2}{0.35}{0}
  \line{0}{1}{0}{0}{0}{0}
  \line{1}{1}{0}{1}{0.65}{0}
  \line{0}{0.65}{0}{1}{0.65}{0}
  \line{2}{1}{0}{3}{0}{0}
  \point{0}{0}{0}{white}
  \point{1}{0}{0}{white}
  \point{2}{0}{0}{white}
  \point{3}{0}{0}{white}
  \point{0}{1}{0}{white}
  \point{1}{1}{0}{white}
  \point{2}{1}{0}{white}
  \label{0}{-0.35}{0}{\Huge 1}
  \label{1}{-0.35}{0}{\Huge 3}
  \label{2}{-0.35}{0}{\Huge 3}
  \label{3}{-0.35}{0}{\Huge 2}
  \label{0}{1.35}{0}{\Huge 1}
  \label{1}{1.35}{0}{\Huge 1}
  \label{2}{1.35}{0}{\Huge 2}
}}.
\]
\end{example}

In the following, we identify partitions $p \in \PP$ with a 
corresponding representation $(k, \ell, b)$. In particular, we write $p := (k, \ell, b) \in \PP$
and denote with $\abs{p} := k + \ell$ the \emph{size} of $p$.

Note that two different representations in the sense of \Cref{def:part-struct}
can correspond to the same partition as string diagram. However, this is exactly is the case if
both partitions are equivalent in the following sense.

\begin{definition}[Equivalence]
Let $p := (k, \ell, a), q := (m, n, b) \in \PP$ be two partitions. Then 
$p$ and $q$ are \emph{equivalent} if 
$k = m$, $\ell = n$ and there exists a bijection 
\[
  \phi \colon \{ a_1, \dots, a_{k+\ell}\} \to \{b_1, \dots, b_{k+\ell}\}
\]
such that $\phi(a_{i}) = b_{i}$ for all $1 \leq i \leq k + \ell$.
\end{definition}

While the representation of a partition is not unique, it is always possible to normalize its block structure
such that the following condition is satisfied.

\begin{definition}\label{def:part-normalized}
Let $p := (k, \ell, b) \in \PP$ be a partition. Then $p$ is \emph{normalized}
if $b_1 = 1$ and $a_j \leq a_i + 1$ for all $1 \leq i < j \leq k$.
\end{definition}

The following algorithm computes the normalization of a given partition
using a map data-structure. In particular, this allows us to efficiently compare partitions since
two partitions are equivalent if and only if their normalized versions agree.

\begin{algorithm}[!htb]%
\caption{Normalize $p$.}%
\label{algo:normalize}%
\begin{flushleft}
\hspace*{\algorithmicindent}\textbf{input:} $p = (k, \ell, b)$ \\
\hspace*{\algorithmicindent}\textbf{output:} equivalent partition satisfying \Cref{def:part-normalized}
\end{flushleft}
\begin{algorithmic}[1]
\State $m \gets 1$
\State $r \gets \text{empty map structure}$
\State $a \gets (b_1, \dots, b_{k + \ell})$
\For{$i = 1, \dots, k + \ell$}
  \If{\text{$r(a_i)$ not defined}}
    \State $r(a_i) \gets m$
    \State $m \gets m + 1$
  \EndIf
  \State $a_i \gets r(a_i)$
\EndFor
\State \Return $(k, \ell, a)$
\end{algorithmic}
\end{algorithm}

\begin{proposition}
Let $p \in \PP$ be a partition. Then \Cref{algo:normalize} computes an 
equivalent partition that is normalized in the sense of \Cref{def:part-normalized}
using $\OO(\abs{p})$ accesses to a map structure.
\end{proposition}
\begin{proof}
The result is equivalent to $p$ since it has again $k$ upper points, $\ell$ lower points and its block structure
is related to $p$ via the bijection $r$. Moreover, the result is normalized because 
$r(a_1) = 1$ and $r(a_i) \geq r(a_j) + 1$ for all $1 \leq j < i \leq k + \ell$ by construction of the map $r$.
Finally, \Cref{algo:normalize} requires $2\cdot(k + \ell) = \OO(\abs{p})$ accesses to the map structure
in the worst case.
\end{proof}

\begin{remark}
By implementing the map structure in \Cref{algo:normalize} as balanced search tree, we
obtain a worst case time complexity of $\OO(n \log n)$ for $n := \abs{p}$. This can be
improved to an average complexity of $\OO(n)$ by using a hash map.
\end{remark}

\begin{example}
Consider the partition
\[
  p := \big(2, \, 3, \, (1, 3, 1, 5, 5, 3)\big) \in \PP.
\]
Then $p$ is not normalized since $b_2 = 3 > b_1 + 1$. However, 
\Cref{algo:normalize} computes the map 
\[
  r(1) = 1, \qquad r(3) = 2, \qquad r(5) = 3,
\]
which yields the normalized partition
\[
  \big(2, \, 3, \, (1, 2, 1, 3, 3, 2)\big).
\]
In particular, this shows that $p$ agrees with the partition from \Cref{ex:part-datastruct}.
\end{example}

\subsection{Operations on partitions}

Next, we consider the basic operations on partitions.
In the case of the tensor product and composition, we will need the following
auxiliary procedure to ensure that the blocks of the two input 
partitions $p := (k, \ell, a)$ and $q := (m, n, b)$ are \emph{disjoint}, i.e.\
\[
  \{ a_1, \dots, a_{k+\ell} \} \cap \{ b_1, \dots, b_{n+m} \} =  \emptyset.
\]

\begin{algorithm}[!htb]
\caption{Make blocks of $p$ disjoint from blocks of $q$}\label{algo:mkdisjoint}
\begin{flushleft}
\hspace*{\algorithmicindent} \textbf{input:} $p = (k, \ell, a)$, $q = (m, n, b)$ \\
\hspace*{\algorithmicindent} \textbf{output:} partition equivalent to $p$ and disjoint from $q$
\end{flushleft}
\begin{algorithmic}[1]
\State $t \gets 1 + \max \{ b_i \mid 1 \leq i \leq n + m \}$
\State $c \gets (a_1 + t, \dots, a_{k+\ell} + t)$
\State \Return $(k, \ell, c)$
\end{algorithmic}
\end{algorithm}

\begin{proposition}
Let $p, q \in \PP$ be two partitions. Then \Cref{algo:mkdisjoint} computes a 
partition equivalent to $p$ with blocks disjoint from $q$ in time $\OO(\abs{p} + \abs{q})$.
\end{proposition}
\begin{proof}
The result is equivalent to $p$ since it has the same number of upper and lower points
and their blocks are related by the bijection $x \mapsto x + t$.
Moreover, the resulting blocks are disjoint to $q$ since $c_i = a_i + t > a_i + b_j \geq b_j$ for all $1 \leq i \leq \abs{p}$
and $1 \leq j \leq \abs{q}$.
\end{proof}

In the following, we now assume that the input partitions to our
algorithms always have disjoint blocks. 
This allows us to immediately compute the involution
and tensor products of partitions by rearranging and merging 
the corresponding block structure.

\begin{algorithm}[!htb]
\caption{Involution $p^*$.}\label{algo:involution}
\begin{flushleft}
\hspace*{\algorithmicindent} \textbf{input:} $p = (k, \ell, b)$ \\
\hspace*{\algorithmicindent} \textbf{output:} $p^*$
\end{flushleft}
\begin{algorithmic}[1]
\State $a \gets (b_{k+1}, \dots, b_{k+\ell}, b_1, \dots, b_{k})$
\State \Return $(\ell, k, a)$
\end{algorithmic}
\end{algorithm}

\begin{algorithm}[!htb]
\caption{Tensor product $p \otimes q$.}\label{algo:tensor}
\begin{flushleft}
\hspace*{\algorithmicindent} \textbf{input:} $p = (k, \ell, a)$, $q = (m, n, b)$ \\
\hspace*{\algorithmicindent} \textbf{output:} $p \otimes q$
\end{flushleft}
\begin{algorithmic}[1]
\State $c \gets (a_1, \dots, a_k, b_1, \dots, b_m, a_{k+1}, \dots, a_{k+\ell}, b_{m+1}, \dots, b_{m+n})$
\State \Return $(k + m, \ell + n, c)$
\end{algorithmic}
\end{algorithm}

\begin{proposition}
Let $p, q \in \PP$ be partitions with disjoint blocks. Then
\begin{enumerate}
\item \Cref{algo:involution} computes 
the involution $p^*$ in time $\OO(\abs{p})$,
\item \Cref{algo:tensor} computes 
the tensor product $p \otimes q$ in time $\OO(\abs{p} + \abs{q})$.
\end{enumerate}
\end{proposition}
\begin{proof}
Follows immediately.
\end{proof}

Similarly, it is possible to implement rotations and the vertical reflection of a partition $p \in \PP$ in time 
$\OO(\abs{p})$ by rearranging its points accordingly.

Next, we come to the composition of two partitions. However, 
we first need to introduce an additional data-structure.

\begin{definition}[Union-find structure]
A \emph{union-find structure} $U$ is a data-structure representing 
an equivalence relation on $\N$ that supports the following
operations:
\begin{enumerate}
  \item $\mathit{union}_U(i, j)$: Merge the equivalence classes $[i]$ and $[j]$ for $i, j \in \N$.
  \item $\mathit{find}_U(i)$: Return a canonical representative of the class $[i]$ for $i \in \N$ such that
  \[
    [i] = [j] \ \Longleftrightarrow \ \mathit{find}_U(i) = \mathit{find}_U(j)
    \qquad \forall i, j \in \N.
  \]
\end{enumerate}
Initially, an empty union-find structure represents the equivalence relations given by the singletons $[i] := \{ i \}$ for all $i \in \N$.
\end{definition}

\begin{remark}\label{rem:uf-runtime}
A union-find structure can be efficiently implemented by a disjoint-set forest~\cite{corem09}, which has a time complexity of $\OO(n \log n)$
for $n$ subsequent union and find operations when performing path compression.
This can be further improved to $\OO(n \alpha(n))$ when additionally performing union by rank, where $\alpha$ denotes the inverse Ackerman function~\cite{tarjan84}.
\end{remark}

Using a union-find structure, we can now implement the composition
of two partitions. As before, we assume that both partitions have disjoint blocks.

\begin{algorithm}[!htb]
\caption{Composition $p \cdot q$.}\label{algo:composition}
\begin{flushleft}
\hspace*{\algorithmicindent} \textbf{input:} $p = (\ell, m, a)$, $q = (k, \ell, b)$ \\
\hspace*{\algorithmicindent} \textbf{output:} $p \cdot q$
\end{flushleft}
\begin{algorithmic}[1]
\State $U \gets \text{empty union-find structure}$
\For{$i = 1, \dots, \ell$}
    \State $\mathit{union}_U(a_{i}, b_{k+i})$
\EndFor
\State $c \gets (b_1, \dots, b_{k}, a_{\ell+1}, \dots, a_{\ell+m})$
\For{$i = 1, \dots, k + m$}
    \State $c_i \gets \mathit{find}_U(c_i)$
\EndFor
\State \Return $(k, m, c)$
\end{algorithmic}
\end{algorithm}

\begin{proposition}
Let $p, q \in \PP$ be partitions with disjoint blocks that are composable. Then \Cref{algo:composition} computes 
the composition $p \cdot q$ using $\OO(\abs{p} + \abs{q})$ $\mathit{union}$ and $\mathit{find}$ operations.
\end{proposition}
\begin{proof}
Since the blocks of $p$ and $q$ are disjoint, we can compute the composition by first identify the lower labels of $q$ with 
the corresponding upper labels of $q$ using $\mathit{union}$ operations and then relabel the upper points of $q$ and lower points of $p$ 
with new canonical labels using $\mathit{find}$.
The runtime of the algorithm is determined by the $\ell = \OO(\abs{p} + \abs{q})$ $\mathit{union}$ and 
the $k + m = \OO(\abs{p} + \abs{q})$ $\mathit{find}$ operations. 
\end{proof}

\begin{remark}
When implementing the union-find structure as disjoint-set forest 
using path-compression (see \Cref{rem:uf-runtime}), \Cref{algo:composition} has a worst case time complexity of $\OO(n \log n)$ for $n := \abs{p} + \abs{q}$.
Note that it is also possible to achieve a time complexity of $\OO(n)$ by first constructing a 
graph with the $\OO(n)$ vertices $\{a_1, \dots, a_{\ell+m}, b_1, \dots, b_{k+\ell}\}$ and
adding the $\OO(n)$ edges $(a_i, b_i)$ for $1 \leq i \leq \ell$ during the $\mathit{union}$ operations. Then, we can use a depth-first search to 
compute all connected components in $\OO(n)$ and select a canonical vertex per component as final label that is returned during $\mathit{find}$. 
\end{remark}

\subsection{Generalizations of partitions}

Next, we discuss how the previous algorithms
can be extended to colored and spatial partitions.
In the case of colored partitions, we can extend \Cref{def:part-struct}
by additionally storing strings of upper and lower colors.

\begin{definition}[Colored partition structure]\label{def:struct-colored-part}
A colored partition can be represented by 
a tuple $(x, y, p)$, where $x \in \{\circ,\bullet\}^k$
and $y \in \{\circ,\bullet\}^\ell$ are the upper and lower colors of an underlying partition $p \in \PP(k, \ell)$.
\end{definition}

Given a colored partition of this form, all previous algorithms can still be applied by restricting to the underlying partition $p$.
Only in the precondition of the composition, it no longer suffices to check the number of points, but we need to compare their complete color strings.

In the case of spatial partitions, we can use the bijection
\[
   \Phi_{k,\ell}^{(m)} \colon \{ 1, \ldots k + \ell \} \times \{1, \ldots m\} \to \{ 1, \ldots, m \cdot (k + \ell) \}
\]
that identifies a point $(i, j)$ with the point $m \cdot (i-1) + j$. It can be extended to a bijection on partitions 
\[
    \Phi_{k,\ell}^{(m)} \colon \PP^{(m)}(k, \ell)
    \to \PP^{(1)}({mk}, {m\ell})
\]
by rearranging the points accordingly, e.g.\
\[
    {\partition[3d]{
      \line{0}{0}{0}{0}{1}{0}
      \line{0}{0}{1}{0}{1}{2}
      \line{0}{0}{2}{0}{1}{1}
      \line{1}{0}{0}{1}{0.35}{0}
      \line{1}{0}{1}{1}{0.35}{1}
      \line{1}{0.35}{0}{1}{0.35}{1}
      \line{1}{0}{2}{1}{0.35}{2}
      \point{0}{0}{0}{white}
      \point{0}{0}{1}{white}
      \point{0}{0}{2}{white}
      \point{1}{0}{0}{white}
      \point{1}{0}{1}{white}
      \point{1}{0}{2}{white}
      \point{0}{1}{0}{white}
      \point{0}{1}{1}{white}
      \point{0}{1}{2}{white}
    }} 
    \, \in \PP^{(3)}(1, 2) 
    \quad \longrightarrow \quad 
    {\partition[2d]{
      \line{0}{0}{0}{0}{1}{0}
      \line{1}{0}{0}{2}{1}{0}
      \line{2}{0}{0}{1}{1}{0}
      \line{3}{0}{0}{3}{0.35}{0}
      \line{4}{0}{1}{4}{0.35}{0}
      \line{3}{0.35}{0}{4}{0.35}{0}
      \line{5}{0}{0}{5}{0.35}{0}
      \point{0}{0}{0}{white}
      \point{1}{0}{0}{white}
      \point{2}{0}{0}{white}
      \point{3}{0}{0}{white}
      \point{4}{0}{0}{white}
      \point{5}{0}{0}{white}
      \point{0}{1}{0}{white}
      \point{1}{1}{0}{white}
      \point{2}{1}{0}{white}
    }} 
    \, \in \PP^{(1)}(3, 6). 
\] 
Thus, we can reduce the implementation of spatial partitions to the case of classical partitions on one level.

\begin{definition}[Spatial partition structure]\label{def:struct-spatial-part}
  A spatial partition $p \in \PP^{(m)}(k, \ell)$ can be represented by 
  a pair $(m, q)$, where $m \in \N$ denotes the number of levels and $q := \phi^{(m)}_{k,l}(p) \in \PP$.
\end{definition}

Since $\phi^{(m)}_{k,l}$ is bijective, we can represent any spatial partition uniquely (up to equivalence) by the pair in the previous definition.
Moreover, $\phi^{(m)}_{k,l}$ is functorial in the sense that it respects the composition, tensor product and involution of spatial partitions. Thus, we can reduce all operations to the underlying partition $q$ after verifying that the number of levels $m$ agrees.
See also~\cite{faross24c} for further information on this bijection and a proof of its functoriality. 

\section{Undecidability of category problems}\label{sec:undeciability}

Consider a category of partitions $\CC$.
Then the following algorithmic problems appear naturally when studying $\CC$ or its corresponding quantum group:
\begin{enumerate}
  \item Is $p \in \CC$ for a partition $p \in \PP$?
  \item What is the number of partitions in $\CC(k, \ell)$?
\end{enumerate}

While these problems can be easily solved for the categories of partitions in \Cref{ex:categories}, it is not clear how to solve these problems in general if for example the category $\CC$ is given in terms of generators as in \Cref{ex:categories-gens}.
In this section, we show that there exists a recursively enumerable category of partitions $\CC$ for which the previous problems are undecidable, i.e.\ there 
exists no Turing machine that solves any of the previous problems for all possible input values.

Here, a category of partitions $\CC$ is called \emph{recursively enumerable} if there exists a Turing machine that enumerates all partitions $p \in \CC$. Note that if we drop this requirement, then the previous statement
follows directly from~\cite{raum16b} where it is shown that the set of all categories
of partitions is uncountable, while the set of all Turing machines is countable. However, in this case, it might not even be possible
to explicitly describe the category $\CC$ or enumerate its elements.   

\subsection{Varieties of groups}

Our proof of the initial statement relies on a result of Raum-Weber~\cite{raum16b} that allows us
to perform a reduction from the identity problem for varieties of groups, which was shown to be undecidable by Kleiman in~\cite{kleiman82}. Thus, we start by
introducing varieties of groups and their corresponding identity problem. See~\cite{neumann67} for 
more information on the general theory of varieties of groups and~\cite{kharlampovich95} for algorithmic problems in
this setting.

In the following, all groups are written multiplicative with $1$ as the identity element. Further, we denote by $\F_\infty$ the free group on the countably many generators $x_1, x_2, \dots$.
If $w \in \F_\infty$ is a word in the variables $x_1, \dots, x_n$, then $w$ is an \emph{identity} in a group $G$ if $w(g_1, \dots, g_n) = 1$ for all $g_1, \dots, g_n \in G$, where $w(g_1, \dots, g_n) \in G$ denotes the element obtained by substituting each generator $x_i$ with $g_i$.
Using this notation, we can now introduce varieties of groups.

\begin{definition}[Variety of groups]
Let $W \subseteq \F_\infty$. Then the \emph{variety of groups} $\VV(W)$
defined by $W$ is the class of all groups $G$ such that $w$ is an identity in $G$ for all $w \in W$.
\end{definition}

If $\VV$ is a variety of groups, then a word $w \in \F_\infty$ is an \emph{identity}
in $\VV$ if it is an identity in all $G \in \VV$. 
Moreover, we say that a variety of groups $\VV(W)$ is \emph{finitely based} if the set $W$ is finite. 
In this setting, Kleiman~\cite{kleiman82} showed that
the identity problem for varieties of groups is undecidable in the following sense.

\begin{proposition}\label{prop:ident-problem-varieties}
There exists a finitely based variety of groups $\VV(W)$
such that the problem of determining if a word $w \in \F_\infty$ is an identity in $\VV(W)$
is undecidable.
\end{proposition}
\begin{proof}
See~\cite{kleiman82} or the proof sketch in~\cite{kharlampovich95}.
An alternate proof can be found in~\cite{aivazyan88}. 
\end{proof}

In addition to the previous proposition, we need an additional fact from the theory of varieties of groups that establishes a bijection 
between variety of groups and fully invariant subgroups of $\F_\infty$. Here, 
a subgroup $H \subseteq G$ is \emph{fully invariant} if $\phi(H) \subseteq H$ 
for every endomorphism $\phi \in \End(G)$.

\begin{proposition}\label{prop:varieties-fully-inv}
There exists a bijection between varieties of groups $\VV$ and fully invariant subgroups of $H \subseteq \F_\infty$
such that
\[
    H = \{ w \in \F_\infty \mid \text{$w$ is an identity in $\VV$} \}.
\]
\end{proposition}
\begin{proof}
See \cite[Theorem 14.31]{neumann67}.
\end{proof}

\begin{remark}\label{rem:var-rec-enum}
Consider a variety of groups $\VV(W)$.
Then~\cite[12.31]{neumann67} shows that the corresponding subgroup $H \subseteq \F_\infty$
in the previous proposition consists of all words that 
can be obtained by applying a finite number of the following operations to elements of $W$:
\begin{align*}
    v &\mapsto v^{-1}, \\ 
  (v, w) &\mapsto vw, \\
  v &\mapsto v(u_1, \dots, u_k) \quad \forall u_1, \dots, u_k \in \F_\infty.  
\end{align*}
In particular, this shows that if a variety of groups is finitely based, 
then the corresponding subgroup $H$ is recursively enumerable.
\end{remark}

\subsection{Hyperoctahedral categories of partitions}

In~\cite{raum16b}, Raum-Weber constructed an embedding of varieties of groups into hyperoctahedral categories of partitions, i.e.\ categories of partitions containing ${\partition[2d]{
  \line{2}{1}{0}{1}{0}{0}
  \line{0}{1}{0}{0}{0.65}{0}
  \line{1}{1}{0}{1}{0.65}{0}
  \line{0}{0.65}{0}{1}{0.65}{0}
  \line{2}{0}{0}{2}{0.35}{0}
  \line{3}{0}{0}{3}{0.35}{0}
  \line{2}{0.35}{0}{3}{0.35}{0}
  \line{1}{0.65}{0}{2}{0.35}{0}
  \point{0}{1}{0}{white}
  \point{1}{1}{0}{white}
  \point{2}{1}{0}{white}
  \point{1}{0}{0}{white}
  \point{2}{0}{0}{white}
  \point{3}{0}{0}{white}
}}$. It is based on a bijection between of these categories and so-called $sS_\infty$-invariant normal subgroups of $\Z_2^{\ast \infty}$.
Here, $\Z_2^{*\infty}$ denotes the countably infinite free product of $\Z_2$
with canonical generators $a_1, a_2, \dots$ in each factor. Moreover, a subgroup $H \subseteq \Z_2^{*\infty}$ is called \emph{$sS_\infty$-invariant} if 
it is invariant under all endomorphisms $\phi \in \End(\Z_2^{*\infty})$ of the form
$\phi(a_i) = a_{f(i)}$ for a function $f \colon \N \to \N$.

\begin{proposition}\label{prop:sSinf-norm-categories}
There exists a bijection between $sS_\infty$-invariant normal subgroups $H \subseteq \Z_2^{*\infty}$ and hyperoctahedral categories of partitions $\CC$ such that 
\[
  \CC(0, n) = \{ p(i_1, \dots, i_n) \mid i_1, \dots, i_n \in \N^n,\, a_{i_1} \cdots a_{i_n} \in H \}
\]
for all $n \in \N$, where $p(i_1, \dots, i_n)$ denotes the partition with $n$ lower points and two points $k, \ell \in \{1, \dots, n\}$ in 
the same block if and only if $i_k = i_\ell$.
\end{proposition}
\begin{proof}
See~\cite[Theorem 4.4 \& 4.6]{raum16b}.
\end{proof}

Now, consider the subgroup $E \subseteq \Z_2^{*\infty}$ of all words of even length given by 
\[
  E := \{ a_{i_1} \cdots a_{i_{2k}} \mid k \in \N, \, i_1, \dots, i_{2k} \in \N \}.
\]
Then one verifies that $E$ is isomorphic $\F_\infty$ via the map
\[
  \Phi \colon \F_\infty \to E, \quad x_n \mapsto a_1 a_{n+1} \quad \forall n \in \N.
\]
The following proposition now shows that fully invariant subgroups
of $E$ yield $sS_\infty$-invariant normal subgroups of $\Z_2^{*\infty}$
under this isomorphism. In particular, this proposition corrects the original version in~\cite{raum16b}, which does not hold in full generality.

\begin{proposition}\label{prop:sSinf-norm}
The subgroup $E \subseteq \Z_2^{*\infty}$ is $sS_\infty$-invariant and normal.
In particular, every fully invariant subgroup $H \subseteq E$ is 
$sS_\infty$-invariant and normal in $\Z_2^{*\infty}$.
\end{proposition}
\begin{proof}
Let $a = a_{i_1} \dots a_{i_{2k}} \in E$ and $\phi \in \End(\Z_2^{*\infty})$
of the form $\phi(a_i) = a_{f(i_1)} \dots a_{f(i)}$ for $f \colon \N \to \N$. Then 
$\phi(a) = a_{f(i_1)} \dots a_{f(i_{2k})} \in E$ such that $E$ is $sS_\infty$-invariant.
Similarly, let $\phi \in \End(\Z_2^{*\infty})$ be an inner automorphism of the form $\phi(x) = b x b^{-1}$ for $b \in \Z_2^{*\infty}$.
Then $\phi(a) = bab^{-1} \in E$ since both $a$ and $bb^{-1}$ have even length.
Since a subgroup is normal if and only if it is invariant under inner automorphism, it follows that 
$E$ is a normal subgroup. 

In particular, this shows that any of the previous endomorphisms restricts to 
an endomorphism $\phi|_E \in \End(E)$. Thus, any fully invariant subgroup $H \subseteq E$ 
is also invariant under these endomorphisms and therefore a $sS_\infty$-invariant and normal subgroup of $\Z_2^{*\infty}$.
\end{proof}

Thus, by combining the previous proposition with 
\Cref{prop:varieties-fully-inv} and \Cref{prop:sSinf-norm-categories}, we obtain an embedding of varieties
of groups into the set of category of partitions.

In particular, we can associate to each word $w \in \F_\infty$ a partition $p(w)$
by first applying isomorphism $\F_\infty \cong E$ to obtain
an element $a_{i_1} \dots a_{i_n} \in E$ and then constructing the partition
$p(i_1, \dots, i_n)$ as in \Cref{prop:sSinf-norm-categories}.

\begin{example}
Consider the word $w := x_{1} x_{2} x_{1}^{-1} x_{3}^{-1} x_{2} x_{3} \in \F_\infty$ that corresponds to 
the element
\[
     a_1 a_2 a_1 a_3 a_2 a_1 a_4 a_1 a_1 a_3 a_1 a_4 \in \Z_2^{*\infty}
\]
under the isomorphism $\F_\infty \cong E$. 
Then the corresponding partition $p(w)$ is given by 
\[
  p(w) \ = \
  {\partition[2d]{
  \line{0}{0}{0}{0}{0.5}{0}
  \line{1}{0}{0}{1}{1}{0}
  \line{2}{0}{0}{2}{0.5}{0}
  \line{3}{0}{0}{3}{1.5}{0}
  \line{4}{0}{0}{4}{1}{0}
  \line{5}{0}{0}{5}{0.5}{0}
  \line{6}{0}{0}{6}{1}{0}
  \line{7}{0}{0}{7}{0.5}{0}
  \line{8}{0}{0}{8}{0.5}{0}
  \line{9}{0}{0}{9}{1.5}{0}
  \line{10}{0}{0}{10}{0.5}{0}
  \line{11}{0}{0}{11}{1}{0}
  \line{0}{0.5}{0}{10}{0.5}{0}
  \line{1}{1}{0}{4}{1}{0}
  \line{6}{1}{0}{11}{1}{0}
  \line{3}{1.5}{0}{9}{1.5}{0}
  \point{0}{0}{0}{white}
  \point{1}{0}{0}{white}
  \point{2}{0}{0}{white}
  \point{3}{0}{0}{white}
  \point{4}{0}{0}{white}
  \point{5}{0}{0}{white}
  \point{6}{0}{0}{white}
  \point{7}{0}{0}{white}
  \point{8}{0}{0}{white}
  \point{9}{0}{0}{white}
  \point{10}{0}{0}{white}
  \point{11}{0}{0}{white}
  \label{0}{-0.35}{0}{\Huge 1}
  \label{1}{-0.35}{0}{\Huge 2}
  \label{2}{-0.35}{0}{\Huge 1}
  \label{3}{-0.35}{0}{\Huge 3}
  \label{4}{-0.35}{0}{\Huge 2}
  \label{5}{-0.35}{0}{\Huge 1}
  \label{6}{-0.35}{0}{\Huge 4}
  \label{7}{-0.35}{0}{\Huge 1}
  \label{8}{-0.35}{0}{\Huge 1}
  \label{9}{-0.35}{0}{\Huge 3}
  \label{10}{-0.35}{0}{\Huge 1}
  \label{11}{-0.35}{0}{\Huge 4}
}}.
\]
\end{example}

\subsection{Proof of the main theorem}

Using the previous embedding of varieties of groups into categories of partitions, 
we can now reduce the identity problem for varieties to the membership
problem for categories of partitions.

\begin{lemma}\label{lem:embedd-varieties}
Let $\VV$ be a finitely based variety of groups. Then there 
exists a recursively enumerable category of partitions $\CC$ such that 
\[
    \text{$w$ is an identity in $\VV$}
    \quad 
    \Longleftrightarrow
    \quad 
    p(w) \in \CC
    \qquad 
    \forall w \in \F_\infty.
\]
\end{lemma}
\begin{proof}
By \Cref{prop:varieties-fully-inv}, $\VV$ corresponds to a fully invariant subgroup $H \subseteq \F_\infty$ that 
can be identified with a $sS_\infty$-invariant normal subgroup of $\Z_2^{*\infty}$ using the isomorphism $\F_\infty \cong E$
and \Cref{prop:sSinf-norm}. Moreover, this subgroup corresponds to a category of partitions $\CC$ by \Cref{prop:sSinf-norm-categories}.
Following this identification, every identity $w \in H$ yields a partition $p(w) \in \CC$.
Conversely, if $p(w) \in \CC$ for $w \in \F_\infty$, then $w \in H$ since the mapping in \Cref{prop:sSinf-norm-categories} is bijective.

It remains to show that the category $\CC$ is recursively enumerable. 
Let $w \in H$ and denote with $a(w)$ the correspond element in $E$. Then, the sets
\[
  \{ (i_1, \dots, i_n) \mid n \in \N, \, i_1, \dots, i_n \in \N^n,\, a_{i_1} \cdots a_{i_n} = a(w) \}
\]
can be enumerated by generating all words $a_{i_1} \cdots a_{i_n}$ and checking the equality.
Since the set $H$ is recursively enumerable by \Cref{rem:var-rec-enum}, we can interleave the elements of the previous sets such that
their union and the set $\bigcup_{n \in \N} \CC(0, n)$ are also recursively enumerable.
Because partitions in $\CC(k, \ell)$ are exactly rotate version of partitions in 
$\CC(0, k+\ell)$, it follows that the category $\CC$ is recursively enumerable.
\end{proof}

Using the previous lemma, we can finally transfer the undecidability of the identity problem 
for varieties of groups to the membership problem in categories of partitions.

\begin{theorem}\label{thm:membership-undecidable}
There exists a recursively enumerable category of partitions $\CC$ 
such that the problem of determining if $p \in \CC$ for a partition $p \in \PP$ is undecidable.
\end{theorem}
\begin{proof}
Let $\VV$ be the variety of group from \Cref{prop:ident-problem-varieties} for which the word problem
is undecidable and let $\CC$ be the category of partitions obtained from \Cref{lem:embedd-varieties}.
Then any Turing machine deciding if $p \in \CC$ for arbitrary $p \in \PP$ could be immediately used to 
decide the identity problem for $\VV$. Thus, the membership problem for $\CC$ also has to be undecidable.
\end{proof}

Additionally, we obtain as an immediate consequence that the problem of counting
partitions of a given size is undecidable too.

\begin{corollary}\label{corr:counting-undecidable}
There exists a recursively enumerable category of partitions $\CC$ 
such that there exists no Turing machine that computes $\abs{\CC(k, \ell)}$
for all $k, \ell \in \N$.
\end{corollary}
\begin{proof}
Let $\CC$ be the category of partitions from \Cref{thm:membership-undecidable}. Assume that there exists a Turing machine $T$ that computes $\abs{\CC(\cdot, \cdot)}$
and that we are given a partition $p \in \PP$. We show that $T$ allows us to determine
if $p \in \CC$, contradicting \Cref{thm:membership-undecidable}.
First, we use $T$ to compute the number of partitions $m := \abs{\CC(k, \ell)}$, where 
$k, \ell \in \N$ are the number of upper and lower points of $p$. Next, we start enumerating
all possible partitions in $\CC$ until we find $m$ distinct partitions $p_1, \dots, p_m \in \CC(k, \ell)$ after a finite time. 
Then, we stop the enumeration and can decide if $p \in \{p_1, \dots, p_m\} = \CC(k, \ell)$.
\end{proof}
 
\section{Conclusion and open questions}\label{sec:conclusion}

In this paper, we presented efficient algorithms and data-structures that allow us to perform any category operation on partitions in quasi-linear time, which is near optimal since for all operations the output size already grows linear in the input size. Moreover, we provide a concrete implementation of these algorithms in the computer algebra system OSCAR, which enables the study of categories of partitions and their corresponding quantum groups via computational methods in future work.

In the second part, we constructed a category of partitions $\CC$ for which the natural problems of deciding the membership of a partition and counting partitions of a given size are algorithmically undecidable.
While the category $\CC$ is constructed from a finitely based variety of groups, we were only able to show that the resulting category is recursively enumerable. Thus, it remains open if $\CC$ is again finitely generated. Using the previous approach, one would have to show that if $H \subseteq \Z_2^{*\infty}$
is finitely generated as fully invariant subgroup of $E$, then $H$ is also finitely generated 
as $sS_\infty$-invariant normal subgroup of $\Z_2^{*\infty}$.

\bibliographystyle{amsplain}
\bibliography{article}

\appendix

\section{OSCAR implementation}\label{sec:oscar-impl}
In addition to the pseudo-code in \Cref{sec:algos-structs}, we provide an implementation of partitions and their operations in OSCAR~\cite{OSCAR}, 
an open-source computer algebra system written in Julia that combines the capabilities of various systems like GAP, Polymake and Singular. The concrete implementation details can be found on the official GitHub repository of OSCAR:
\begin{center}
    \url{https://github.com/oscar-system/Oscar.jl/tree/master/experimental/SetPartitions}
\end{center}

In the following, we give a short demo presenting the basic functionality of our implementation. The main data-structure is given by \texttt{SetPartition}, which implements partitions and is accessible through the \texttt{set\_partition} constructor. In contrast to \Cref{def:part-struct}, a partition is constructed from two vectors, where the first vector represents the $k$ upper points and the second vector the $\ell$ lower point. 

\begin{lstlisting}
In: using Oscar
In: set_partition([2, 4], [4, 99])
Out: SetPartition([1, 2], [2, 3])
\end{lstlisting}

Note that the resulting partition is automatically normalized as described in \Cref{algo:normalize}. Thus, it is possible to directly compare two partitions or store them in a hash map. To apply basic operations on partitions, we can use the following functions.

\begin{lstlisting}
In: p = set_partition([1, 2], [2, 1])
In: q = set_partition([1, 1], [1])
In: tensor_product(p, q)
Out: SetPartition([1, 2, 3, 3], [2, 1, 3])
\end{lstlisting}

\begin{lstlisting}
In: p = set_partition([1, 2, 2], [1, 1, 3])
In: involution(p)
Out: SetPartition([1, 1, 2], [1, 3, 3])
\end{lstlisting}

\begin{lstlisting}
In: p = set_partition([1, 2, 2], [1, 2])
In: q = set_partition([1], [2, 2, 1])
In: compose(p, q)
Out: SetPartition([1], [1, 1])
\end{lstlisting}

Similarly, the vertical reflections and rotations can be computed using functions like \texttt{reflect\_vertical} and \texttt{rotate\_top\_left}.
Moreover, we also support colored partitions and spatial partitions via the structures \texttt{ColoredPartition} and \texttt{SpatialPartition}.

In addition to the previous algorithms and data-structures, we provide a function \texttt{construct\_category}. The first argument is a vector of generating partitions and the second argument is a bound $n$. The functions then constructs all partitions of size $n$ that can be constructed from the given generating set and base partitions without the size of intermediate partitions exceeding $n$. The following example corresponds to the category of all non-crossing partitions from \Cref{ex:categories-gens}.

\begin{lstlisting}
In: p = set_partition([1], [1, 1])
In: q = set_partition([1], [1])
In: r = set_partition([], [1, 1])
In: length(construct_category([p, q, r], 6))
Out: 924
\end{lstlisting}

Since the size of the output grows exponentially, the function might take a while for larger bounds. Note that in this example, the algorithm constructs all non-crossing partitions of size $6$ since their number is given by $7 \cdot C_6 = 924$, where $C_k$ denotes the $k$-th Catalan number. However, this is not guaranteed in general since one might need to first construct larger partitions to obtain again smaller partition.

\end{document}